\documentclass[conference]{IEEEtran}
\twocolumn
\usepackage[english]{babel}
\usepackage{amsmath,amsthm}
\usepackage{amsfonts}
\usepackage[final,hyperfootnotes=false,hidelinks]{hyperref}
\usepackage[final]{graphicx}
\usepackage{epstopdf}
\usepackage{tikz}
\usepackage{pgfplots}
\usepackage{subfig}
\usepackage{flushend}
\usepackage{cprotect}

\usepackage{tikz}
\usetikzlibrary{arrows,positioning,fit,backgrounds}
\usetikzlibrary{calc}

\usetikzlibrary{arrows,positioning,fit,backgrounds,shapes}
\usetikzlibrary{calc}
\newtheorem{thm}{Theorem}

\newtheorem{lem}[thm]{Lemma}

\theoremstyle{definition}

\theoremstyle{remark}
\newtheorem{rem}[thm]{Remark}

\newcommand{\mb}{\mathbf}

\begin{document}

\title{One-Shot Mutual Covering Lemma and Marton's Inner Bound with a Common Message}%
\author{\IEEEauthorblockN{Jingbo Liu~~~~~~~~Paul Cuff ~~~~~~~~~Sergio Verd\'{u}}
\IEEEauthorblockA{Dept. of Electrical Eng., Princeton University, NJ 08544\\
\{jingbo,cuff,verdu\}@princeton.edu}}%
\maketitle

\begin{abstract}
By developing one-shot mutual covering lemmas, we derive a one-shot achievability bound for broadcast with a common message which recovers Marton's inner bound (with three auxiliary random variables) in the i.i.d.~case. The encoder employed is deterministic. Relationship between the mutual covering lemma and a new type of channel resolvability problem is discussed.
\end{abstract}

\section{Introduction}
While many problems in network information theory have been successfully solved in the discrete memoryless case using either the method of types or weak/strong typicality (e.g., \cite{el2011network} or \cite{csiszar1981information}), there is a gap between the i.i.d.~assumptions underlying these methods and the nature of sources and channels arising from real world applications. Recent works (e.g.~\cite{polyanskiy2010channel}\cite{Kostina}) have developed new methods to derive tight one-shot achievability bounds for specialized problems. Meanwhile, a natural question is whether there exist general techniques to attach non-asymptotic fundamental limits in multiuser information theory.
By developing one-shot versions of covering and packing lemmas, \cite{verdu2012non} successfully obtained one-shot achievability bounds for various problems (multi-access, Slepian-Wolf, Gelfand-Pinsker, Wyner-Ziv, Ahlsede-K\"{o}rner, and broadcast of private messages) which lead to such non-asymptotic bounds and recover known results in the i.i.d.~case. However the proof of Marton's inner bound of broadcast without public/common messages (the two auxiliary version, \cite[Theorem~8.3]{el2011network}) proceeded by showing the achievability of each corner point, which requires time sharing to recover the full rate region in the i.i.d.~case.

In this paper we develop a one-shot mutual covering lemma so that a one-shot version of Marton's inner bound with a common message (the three auxiliary version originally due to Liang et~al.~\cite{liang2007rate}; see also \cite[Theorem~8.4]{el2011network}) can be obtained without time sharing.

Time sharing may not be satisfactory since it is meaningless in the one-shot case. This is keenly noted by the authors of \cite{yassaee2013technique}, who also observed that the mutual covering lemma~\cite[Lemma~8.1]{el2011network}, a technique for avoiding time sharing in the i.i.d.~case, does not seem to have a one-shot counterpart:
\begin{quote}
\emph{... This is helpful because no one-shot extension of the mutual covering and packing lemma exists... For this reason Verd\'{u}'s result seems to be weaker than ours.}
\end{quote}
The present paper gives a single-shot mutual covering lemma, thereby filling the gap noted by the above remark. We will provide two derivations, one of them is based on a recent generalization of channel resolvability \cite{liu2015resolvability}.

Partly motivated by the desirability of achievability bounds without time sharing, \cite{yassaee2013technique} (see also \cite{yassaee}) proposed a new technique for deriving one-shot achievability bounds using stochastic encoders, thus avoiding the mutual covering lemma and yielding a one-shot bound which recovers the two auxiliary-variable version of Marton's inner bound. However that version of Marton's bound is known to not be tight, whereas the three auxiliary-variable version (Liang-Kramer bound) is still a candidate for the capacity region \cite[Section~8.4]{el2011network}. In \cite{yassaee} a one-shot achievability bound with three auxiliaries is stated without a proof\footnote{The authors of \cite{yassaee} announced that the proof will be included in their future draft.}, which is not easily comparable with our bound. However the main probability terms from the two bounds are the same, so they are equivalent in the second order rate analysis.

\section{One-shot Mutual Covering Lemma}\label{sec1}
We develop one-shot mutual covering lemma(s) in this section, which require a different proof idea than the i.i.d.~case \cite[Appendix~8A]{el2011network} because the empirical distribution is meaningless in the non-block coding case. The main device in our proof is the introduction of an auxiliary random variable $\tilde{V}$, which can be viewed as being ``typical'' with a random codeword from the $U$ codebook.

To begin with, let us introduce the notations of the relative information
\begin{align}
\imath_{P||Q}(x):=\log\frac{{\rm d}P}{{\rm d}Q}(x)
\end{align}
where $P$ and $Q$ are distributions on a same alphabet $\mathcal{X}$ with $\frac{{\rm d}P}{{\rm d}Q}(x)$ being the Radon-Nikodym derivative, and the information density
\begin{align}
\imath_{U;V}(u;v):=\imath_{P_{V|U=u}||P_V}(v).
\end{align}

\begin{lem}\label{lem1}
Fix $P_{UV}$ and let
\begin{align}
P_{U^MV^L}:=P_U\times \dots\times P_U\times P_V\times\dots\times P_V
\end{align}
then
\begin{align}
&\quad\mathbb{P}[\bigcap_{m=1,l=1}^{M,L}\{(U_m,V_l)\notin \mathcal{F}\}]
\nonumber
\\
&\le
\mathbb{P}[(U,V)\notin\mathcal{F}]
+\mathbb{P}\left[\exp(\imath_{U;V}(U;V))>ML\exp(-\gamma)-\delta\right]
\nonumber
\\
&\quad
+\frac{\min\{M,L\}-1}{\delta}
+e^{-\exp(\gamma)}.
\end{align}
for all $\delta,\gamma>0$ and event $\mathcal{F}$.
\end{lem}
\begin{proof}
Let $(\tilde{V}_1,\dots,\tilde{V}_M)$ be random variables satisfying
\begin{align}
P_{U^M\tilde{V}^M}=P_{UV}\times\dots\times P_{UV},
\end{align}
and define
\begin{align}\label{e_vt}
\tilde{V}:=\tilde{V}_W
\end{align}
where $W$ is equiprobable on $\{1,\dots,M\}$ and independent of $(U^M,\tilde{V}^M)$.
By independence,
\begin{align}
\mathbb{P}[\bigcap_{m=1,l=1}^{M,L}\{(U_m,V_l)\notin \mathcal{F}\}]
=\mathbb{E}\mathbb{P}[\bigcap_{m=1}^{M}\{(U_m,V_1)\notin \mathcal{F}\}|U^M]^L.
\end{align}
Define the functions
\begin{align}
A_{u^M}&:=\mathbb{P}[\bigcup_{m=1}^M\{(u_m,V_1)\in\mathcal{F}\}];
\\
\pi(u^M;v)&:=1\left\{\bigcup_{m=1}^M\{(u_m,v)\in\mathcal{F}\}\right\}
\\
&\quad\cdot1\left\{\imath_{U^M;\tilde{V}}(u^M;v)\le\log L-\gamma\right\}.\label{e_pi2}
\end{align}
The second term in \eqref{e_pi2} will play a role in the change of measure step.
We then obtain
\begin{align}
&\quad[1-A_{U^M}]^L
\\
&=\left[1-\frac{(L\exp(-\gamma)A_{U^M})\exp(\gamma)}{L}\right]^L
\\
&\le \left[1-\frac{\exp(\gamma)}{L}\mathbb{E}[\pi(U^M;\tilde{V})|U^M]\right]^L\label{e8}
\\
&\le 1-\mathbb{E}[\pi(U^M;\tilde{V})|U^M]
+e^{-\exp(\gamma)},\label{e_12}
\end{align}
where
\begin{itemize}
\item \eqref{e8} is from the fact that (by change of measure)
\begin{align}
L\exp(-\gamma)A_{U^M}&\ge
L\exp(-\gamma)\mathbb{E}[\pi(U^M;V_1)|U^M]
\\
&\ge
\mathbb{E}[\pi(U^M;\tilde{V})|U^M].
\label{e_cm}
\end{align}
\item \eqref{e_12} uses the inequality
\begin{align}\label{e_15}
(1-\frac{p\alpha}{M})^M\le 1-p+e^{-\alpha}
\end{align}
for $M,\alpha>0$ and $0\le p\le 1$.
\end{itemize}
Thus
\begin{align}
&\quad\mathbb{E}[1-A_{U^M}]^L
\\
&\le\mathbb{E}[1-\pi(U^M;\tilde{V})]
+e^{-\exp(\gamma)}
\\
&=\frac{1}{M}\sum_{w=1}^M\mathbb{E}[1-\pi(U^M;\tilde{V}_W)|W=w]
+e^{-\exp(\gamma)}\label{e13}
\\
&=\mathbb{E}[1-\pi(U^M;\tilde{V}_1)]
+e^{-\exp(\gamma)}\label{e14}
\\
&\le \mathbb{P}[(U_1,\tilde{V}_1)\notin\mathcal{F},
\textrm{or~}\imath_{U^M;\tilde{V}}(U^M;\tilde{V}_1)>\log L-\gamma]
+e^{-\exp(\gamma)}\label{e_bu}
\\
&\le \mathbb{P}[(U,V)\notin\mathcal{F}]
\nonumber
\\
&\quad+\mathbb{P}\left[\frac{1}{M}\sum_{m=1}^M \frac{{\rm d}P_{V|U=U_m}}{{\rm d}P_V}(\tilde{V}_1)>L\exp(-\gamma)\right]
+e^{-\exp(\gamma)}\label{e16}
\\
&\le \mathbb{P}[(U,V)\notin\mathcal{F}]
+\mathbb{P}\left[\frac{{\rm d}P_{V|U=U_1}}{{\rm d}P_V}(\tilde{V}_1)>ML\exp(-\gamma)-\delta\right]
\nonumber
\\
&\quad
+\mathbb{P}\left[\sum_{m=2}^M \frac{{\rm d}P_{V|U=U_m}}{{\rm d}P_V}(\tilde{V}_1)>\delta\right]
+e^{-\exp(\gamma)}\label{e17}
\\
&\le \mathbb{P}[(U,V)\notin\mathcal{F}]
+\mathbb{P}\left[\frac{{\rm d}P_{V|U=U_1}}{{\rm d}P_V}(\tilde{V}_1)>ML\exp(-\gamma)-\delta\right]
\nonumber
\\
&\quad
+\frac{M-1}{\delta}
+e^{-\exp(\gamma)}\label{e18}
\end{align}
where
\begin{itemize}
  \item \eqref{e14} is because by symmetry, all summands in \eqref{e13} are equal.
  \item \eqref{e16} is because by definition, $P_{U_1\tilde{V}_1}=P_{UV}$.
  \item \eqref{e17} uses the elementary inequality $1\{x+y>a+b\}\le 1\{x>a\}+1\{y>b\}$.
  \item \eqref{e18} uses the Markov inequality.
\end{itemize}
\end{proof}

\begin{rem}
Supposing that $M\le L$, an optimal choice of $(\gamma,\delta)$ is the solution to the optimization problem
\begin{align}\label{opt}
\textrm{minimize}&\quad f(\gamma,\delta):=\frac{M-1}{\delta}+e^{-\exp(\gamma)}
\\
\textrm{subject to}&\quad g(\gamma,\delta):=ML\exp(-\gamma)-\delta=\lambda
\end{align}
for some constant $\lambda\in\mathbb{R}$. Using Lagrange multipliers
we can show that the optimal value of $\delta$ for fixed $\gamma$ is:
\begin{align}\label{e28}
\delta=\sqrt{(M-1)ML}\exp(-\gamma)e^{\frac{1}{2}\exp(\gamma)}.
\end{align}
\end{rem}
%

Note that Lemma~\ref{lem1} recovers the covering lemma in \cite{verdu2012non} when $\min\{M,L\}=1$. Also, Lemma~\ref{lem1} can be weakened to the following simpler bound by setting $\delta=ML(\exp(-\gamma)-\exp(-2\gamma))$.
\begin{lem}\label{lem4}
Under the same assumptions of Lemma~\ref{lem1},
\begin{align}
&\quad\mathbb{P}[\bigcap_{m=1,l=1}^{M,L}\{(U_m,V_l)\notin \mathcal{F}\}]
\nonumber
\\
&\le
\mathbb{P}[(U,V)\notin\mathcal{F}]
+\mathbb{P}\left[\imath_{U;V}(U;V)>\log(ML)-2\gamma\right]
\nonumber
\\
&\quad+\frac{\min\{M,L\}-1}{ML(\exp(-\gamma)-\exp(-2\gamma))}
+e^{-\exp(\gamma)}.
\end{align}
for all $\gamma>0$ and event $\mathcal{F}$.
\end{lem}
\begin{rem}\label{rem5}
In all versions of the mutual covering lemmas given above, the sum of the two probability terms can be strengthened to the probability of the union of two events, if the union bound is not applied to simplify \eqref{e_bu}.
\end{rem}
A conditional version of the mutual covering lemma (c.f.~Lemma 8.1 in \cite{el2011network}) also follows by averaging:
\begin{lem}\label{lem5}
Fix $P_{UST}$ and let
\begin{align}
P_{US^MT^L}:=P_U\times P_{S|U}\times \dots\times P_{S|U}\times P_{T|U}\times\dots\times P_{T|U}
\end{align}
then
\begin{align}
&\quad\mathbb{P}[\bigcap_{m=1,l=1}^{M,L}\{(U,S_m,T_l)\notin \mathcal{F}\}]
\nonumber
\\
&\le
\mathbb{P}[\{(U,S,T)\notin\mathcal{F}\}
\cup\{\imath_{S;T|U}(S;T|U)>\log(ML)-2\gamma\}]
\nonumber
\\
&\quad+\frac{\min\{M,L\}-1}{ML(\exp(-\gamma)-\exp(-2\gamma))}
+e^{-\exp(\gamma)}.
\end{align}
for all $\gamma>0$ and event $\mathcal{F}$.
\end{lem}

\section{Mutual Covering from Resolvability}\label{sec_3}
While the mutual covering lemma in Section~\ref{sec1} suffices the purpose of Section~\ref{sec4}, in this section we provide a simple alternative derivation based on a recent result on resolvability in the excess information metric \cite{liu2015resolvability}, which illustrates the interesting connection between resolvability and the mutual covering lemma\footnote{Indeed, \cite{cuff2012distributed} uses the term ``soft-covering'' for the achievability part of channel resolvability \cite{han1993approximation}, an idea that traces back to Wyner \cite{wyner1975common}.}.
\begin{lem}{\cite{liu2015resolvability}}\label{lem6}
Fix $P_{UV}=P_UP_{V|U}$. Let $\mb{c}=[c_1,\dots,c_M]$ be i.i.d.~according to $P_U$. Define
\begin{align}\label{e_hv}
\hat{P}_V:=\frac{1}{M}\sum_{m=1}^M P_{V|U=c_m}.
\end{align}
Then for any $\lambda>2$,
\begin{align}
&\quad\mathbb{E}\mathbb{P}[\imath_{\hat{P}_V||P_V}(\hat{V})>\log\lambda]
\nonumber
\\
&\le\mathbb{P}\left[\imath_{V;U}(V;U)\ge\log\frac{M\lambda}{2}\right]+\frac{2}{\lambda}
\end{align}
where the expectation is with respect to the codebook realization, $\hat{V}\sim \hat{P}_V$ and $(U,V)\sim P_{UV}$.
\end{lem}
The above lemma follows by setting $\gamma=\lambda$, $\sigma=\epsilon\uparrow\frac{\lambda}{2}$, $\delta\uparrow1$, $Q_{XU}=P_{VU}$, $\pi_X=P_V$, $L=M$ and $P_X=\hat{P}_V$ in \cite[Remark~3]{liu2015resolvability}. We then have
\begin{lem}\label{lem7}
Under the same assumptions of Lemma~\ref{lem1},
\begin{align}
&\quad\mathbb{P}[\bigcap_{m=1,l=1}^{M,L}\{(U_m,V_l)\notin \mathcal{F}\}]
\nonumber
\\
&\le
\mathbb{P}[(U,V)\notin\mathcal{F}]
+\mathbb{P}[\imath_{U;V}(U;V)\ge\log ML-\gamma]
\nonumber
\\
&\quad+\frac{\exp(\gamma)}{\max\{M,L\}}+e^{-\frac{1}{2}\exp(\gamma)}.
\label{e_39}
\end{align}
\end{lem}
\begin{rem}
The bound in Lemma~\ref{lem7} appears similar to and slightly simpler than Lemma~\ref{lem4}. However the sum of the two probabilities in \eqref{e_39} cannot be easily strengthened to the probability of a union (see Remark~\ref{rem5}), which is important in the second order rate analysis.
\end{rem}
\begin{rem}
In both derivations of the one-shot mutual covering lemma, the role of $U$ and $V$ are asymmetric. Moreover, the two methods are not readily extendable to obtain a one-shot version of the multivariate covering lemma \cite[Lemma~8.2]{el1981proof}.
\end{rem}
\begin{proof}[Proof of Lemma~\ref{lem7}]
Assume without loss of generality that $L\ge M$. Define the sets
\begin{align}
\mathcal{F}_{u}&:=\{v:(u,v)\in\mathcal{F}\},
\\
\mathcal{A}_{u^M}&:=\bigcup_{m=1}^M \mathcal{F}_u.
\end{align}
For fixed $\mb{c}=c^M$ in Lemma~\ref{lem6},
\begin{align}
\hat{P}_V(\mathcal{A}_{\mb{c}})-\lambda P_V(\mathcal{A}_{\mb{c}})
&\le \sup_{\mathcal{A}}[\hat{P}_V(\mathcal{A})-\lambda P_V(\mathcal{A})]
\label{e_np1}
\\
&=\mathbb{P}[\frac{{\rm d}\hat{P}_V}{{\rm d}P_V}(\hat{V})>\lambda]
-\lambda\mathbb{P}[\frac{{\rm d}\hat{P}_V}{{\rm d}P_V}(V)>\lambda]
\label{e_43}
\\
&\le\mathbb{P}[\frac{{\rm d}\hat{P}_V}{{\rm d}P_V}(\hat{V})>\lambda]
\label{e_np}
\end{align}
where \eqref{e_43} is from Neyman-Pearson lemma. Thus
\begin{align}
\lambda P_V(\mathcal{A}_{\mb{c}})
&\ge \hat{P}_V(\mathcal{A}_{\mb{c}})-\mathbb{P}[\frac{{\rm d}\hat{P}_V}{{\rm d}P_V}(\hat{V})>\lambda]
\\
&\ge
\frac{1}{M}\sum_{m=1}^MP_{V|U=c_m}(\mathcal{F}_{c_m})-\mathbb{P}[\frac{{\rm d}\hat{P}_V}{{\rm d}P_V}(\hat{V})>\lambda]\label{e_46}
\end{align}
where \eqref{e_46} is from $\mathcal{F}_{c_m}\subseteq \mathcal{A}_{\mb{c}}$.
Denote by $p_{\mb{c}}$ the right hand side of \eqref{e_46}.
Then $p_{\mb{c}}\le1$, and setting $\mb{c}=U^M$ we obtain
\begin{align}
\mathbb{P}[\bigcap_{l=1}^L\{(U_m,V_l)\notin\mathcal{F}\}|U^M]
&=\left[1-P_V(\mathcal{A}_{U^M})\right]^L
\nonumber\\
&\le\left[1-\frac{p_{\mb{c}}\frac{L}{\lambda}}{L}\right]^L
\label{e_47}
\\
&\le 1-p_{\mb{c}}+e^{-\frac{L}{\lambda}}\label{e_48}
\end{align}
where \eqref{e_48} uses \eqref{e_15}. Then the result follows by unconditioning $U^M$ on both sides of \eqref{e_48}, applying \eqref{e_46} and Lemma~\ref{lem6}, and setting $\lambda=2L\exp(-\gamma)$.
\end{proof}
While the derivations of the one-shot mutual covering lemma in Sections~\ref{sec1} and \ref{sec_3} follow different routes, their correspondences can be seen in the following ways:
\begin{enumerate}
  \item The auxiliary random variable $\tilde{V}$, which is the main device in the first proof, has the distribution $\hat{P}_V$ as in \eqref{e_hv} conditioned on $U^M=c^M$.
  \item The change of measure steps \eqref{e_cm} can be related to the Neyman-Pearson lemma \eqref{e_np1}-\eqref{e_np}.
  \item Both derivations relies on the inequality \eqref{e_15} which also appeared in the proof of the standard covering lemma.
\end{enumerate}
Although Lemma~\ref{lem6} implies the one-shot mutual covering lemma, the reverse implication does not seem to follow directly. Thus resolvability in the excess information is a stronger result than the mutual covering lemma.

\section{Inner Bound with a Common Message}\label{sec4}
We prove a single shot version of the asymptotic achievability result of Liang-Kramer \cite[Theorem~5]{liang2007rate} (see also \cite[Theorem~8.4]{el2011network}). This region is equivalent to an inner bound obtained by Gelfand and Pinsker \cite{gel1980capacity} upon optimization (see \cite{liang2011equivalence} or \cite[Remark~8.6]{el2011network}).
\begin{thm}\label{thm1}
Fix arbitrary distributions $P_{Y_1Y_2|X}$, $P_{UST}$, a map $x:\mathcal{U}\times\mathcal{S}\times\mathcal{T}\to\mathcal{X}$, and integers $M_0$, $M_{10}$, $M_{20}$, $N$, $L$, $\hat{N}$ and $\hat{L}$. Set
\begin{align}
M&:=M_0M_{10}M_{20};
\\
M_1&:=M_{10}N;
\\
M_2&:=M_{20}L;
\\
\tilde{N}&:=\hat{N}N;
\\
\tilde{L}&:=\hat{L}L.
\end{align}
Then there exists an $(M_0,M_1,M_2,\epsilon_1,\epsilon_2)$ code with
\begin{align}
&\quad\max\{\epsilon_1,\epsilon_2\}
\nonumber
\\
&\le 2\exp(-\gamma)+e^{-\exp(\gamma)}
\\
&\quad+\mathbb{P}[\{\imath_{US;Y_1}(US;Y_1)\le \log M\tilde{N}+\gamma\}
\nonumber
\\
&\quad\quad \cup\{\imath_{UT;Y_2}(UT;Y_2)\le \log M\tilde{L}+\gamma\}
\nonumber
\\
&\quad\quad \cup\{\imath_{S;Y_1|U}(S;Y_1|U)\le \log \tilde{N}+\gamma\}
\nonumber
\\
&\quad\quad \cup\{\imath_{T;Y_2|U}(T;Y_2|U)\le \log \tilde{L}+\gamma\}
\nonumber
\\
&\quad\quad \cup\{\imath_{S;T|U}(S;T|U)>\log(\hat{N}\hat{L})-2\gamma\}]
\nonumber
\\
&\quad+\frac{\min\{\hat{N},\hat{L}\}-1}{\hat{N}\hat{L}(\exp(-\gamma)-\exp(-2\gamma))}\label{ebd}
\end{align}
where $P_{USTXY_1Y_2}:=P_{UST}P_{X|UST}P_{Y_1Y_2|X}$.
\end{thm}
As in \cite[Theorem~8.4]{el2011network}, the private messages are decomposed into a public part and an individual part:
\begin{align}
W_i=(W_{i0},W_{ii})
\end{align}
where $i=1,2$ and $W_{i0}$ is supposed to be decodable by both users.

\begin{proof}
\begin{itemize}
  \item Codebook Generation: generate
  \begin{align}
  \mathbf{u}=[\textsf{u}_1,\dots,\textsf{u}_M]
  \end{align}
  according to the distribution $P_U^{\otimes M}$. Also for each $1\le i\le M$, generate
  \begin{align}
  \mathbf{s}_i=[\textsf{s}_i(n,\hat{n})]_{1\le n\le N,1\le \hat{n}\le \hat{N}}
  \end{align}
  according to $P_{S|U=\textsf{u}_i}^{\otimes N\hat{N}}$ and
  \begin{align}
  \mathbf{t}_i=[\textsf{t}_i(l,\hat{l})]_{1\le l\le L,1\le \hat{l}\le \hat{L}}
  \end{align}
  according to $P_{T|U=\textsf{u}_i}^{\otimes L\hat{L}}$. (In other words, for each $i$ we construct a codebook similar to \cite[Figure~8.8]{el2011network}, where each small rectangle has size $\hat{L}\times\hat{N}$.)

  \item Encoding: we may assume that the public message is $w_0$ and the private messages for the two users are $(w_{10},a)$ and $(w_{20},b)$ respectively, where
      \begin{align}
      w_0&\in \{1,\dots,M_0\}
      \\
      w_{i0}&\in \{1,\dots,M_{i0}\}
      \\
      a&\in\{1,\dots,N\}
      \\
      b&\in\{1,\dots,L\}
      \end{align}
      for $i=1,2$. Then the index $m=\texttt{m}(w_0,w_{10},w_{20})$ is selected for the lower-layer codebook where $\texttt{m}$ is a bijection between $[M_0]\times[M_{10}]\times[M_{20}]$ and $[M]$.
      The encoder then finds $\hat{a}\in\{1,\dots,\hat{N}\}$ and $\hat{b}\in\{1,\dots,\hat{L}\}$ that minimize
      \begin{align}\label{ezeta}
      \zeta(\textsf{u}_m,\textsf{s}_m(a,\hat{a}),\textsf{t}_m(b,\hat{b}))
      \end{align}
      where
      \begin{align}\label{e_56}
      \zeta(u,s,t):=P_{Y_1Y_2|X=x(u,s,t)}(\mathcal{Y}_{u,s,t})
      \end{align}
      and
      \begin{align}
      \mathcal{Y}_{u,s,t}&:=\{(y_1,y_2)\in\mathcal{Y}_1\times\mathcal{Y}_2:
      \imath_{US;Y_1}(us;y_1)\le \log M\tilde{N}+\gamma
      \nonumber
      \\
      &\textrm{ or }
      \imath_{UT;Y_2}(ut;y_2)\le \log M\tilde{L}+\gamma
      \nonumber
      \\
      &\textrm{ or }
      \imath_{S;Y_1|U}(s;y_1|u)\le \log \hat{N}+\gamma
      \nonumber
      \\
      &\textrm{ or }
      \imath_{T;Y_2|U}(t;y_2|u)\le \log \hat{L}+\gamma
      \}
      \end{align}
      The signal transmitted is then $x(\textsf{u}_m,\textsf{s}_m(a,\hat{a}),\textsf{t}_m(b,\hat{b}))$.

      \item Decoder: assume that $y_1$ and $y_2$ are observed by the two receivers, respectively. The decoder of the first receiver (Decoder~1) finds the unique $m$ such that
      \begin{align}
      \exists d,\hat{d},\quad\imath_{US;Y_1}(\textsf{u}_m\textsf{s}_m(d,\hat{d});y_1)>\log M\tilde{N}+\gamma
      \end{align}
      and set $(m_0,m_{10},m_{20})=\texttt{m}^{-1}(m)$, or declares an error if not possible. Decoder~1 then finds the unique $(c,\hat{c})$ such that
      \begin{align}
      \imath_{S;Y_1|U}(\textsf{s}_m(c,\hat{c});y_1|\textsf{u}_m)>\log \tilde{N}+\gamma
      \end{align}
      or declares an error if not possible. The output public message is then $m_0$ and the output private message is $(m_{10},c)$.

      Decoder~2 performs similar operations as Decoder~1.

      \item Error Analysis: By symmetry of the codebook, we may assume without loss of generality that
        \begin{align}
        w_0=w_{10}=w_{20}=w_{11}=w_{22}=1
        \end{align}
        is sent and $\texttt{m}(1,1,1)=1$. Also it suffices to prove the bound for Decoder~1 only in view of the symmetry of the bound \eqref{ebd}. Assume that $(\hat{a},\hat{b})$ minimizes \eqref{ezeta} (and is selected by the encoder). Decoder~1 fails only if one or more of the following events occur:
        \begin{align}
        \mathcal{E}_1:&\quad \imath_{US;Y_1}(\textsf{u}_1\textsf{s}_1(1,\hat{a});y_1)\le \log M\tilde{N}+\gamma;
        \\
        \mathcal{E}_2:&\quad \exists m\neq 1,d,\hat{d},
        \nonumber
        \\
        &\quad\imath_{US;Y_1}(\textsf{u}_m\textsf{s}_m(d,\hat{d});y_1)>\log M\tilde{N}+\gamma;
        \\
        \mathcal{E}_3:&\quad \imath_{S;Y_1|U}(\textsf{s}_1(1,\hat{a});y_1|\textsf{u}_1)\le \log \tilde{N}+\gamma;
        \\
        \mathcal{E}_4:&\quad \exists n\neq 1,\hat{n},\quad \imath_{S;Y_1|U}(\textsf{s}_1(n,\hat{n})
        ;y_1|\textsf{u}_1)>\log\tilde{N}+\gamma.
        \end{align}

        Denote by $S^*$ and $T^*$ the coefficients selected by the encoder and $Y_1^*$, $Y_2^*$ the corresponding outputs. Averaged over the codebook, we can bound
        \begin{align}
        &\quad\mathbb{P}[\mathcal{E}_2]
        \nonumber
        \\
        &=\mathbb{P}
        [\bigcup_{m\neq1,d,\hat{d}}\{\imath_{US;Y_1}(U_mS_m(d,\hat{d});Y_1^*)
        >\log M\tilde{N}+\gamma\}]
        \\
        &\le (M-1)\tilde{N}\mathbb{E}\mathbb{P}[\imath_{US;Y_1}(U_2S_2(1,1);Y_1^*)
        >\log M\tilde{N}+\gamma|Y_1^*]
        \\
        &\le \frac{M-1}{M\exp(\gamma)}\label{ee2}
        \end{align}
        by change of measure, where the joint probability of $\{U_m,S_m(d,d)\}_{m,d,\hat{d}}$ and $Y_1^*$ satisfies $P_{U_mS_m(d,\hat{d})Y_1^*}=P_{US}\times P_{Y_1^*}$. Similarly
        \begin{align}
        &\quad\mathbb{P}[\mathcal{E}_4]
        \nonumber
        \\
        &=\mathbb{P}[\bigcup_{n\neq1,1\le\hat{n}\le\hat{N}}
        \{\imath_{S;Y_1|U}(S_1(n,\hat{n});Y_1^*|U_1)>\log \tilde{N}+\gamma\}]
        \\
        &\le (N-1)\hat{N}\mathbb{E}\mathbb{P}[\imath_{S;Y_1|U}(S_1(2,1);Y_1^*|U_1)
        \nonumber
        \\
        &\quad\quad\quad\quad\quad\quad\quad>\log \tilde{N}+\gamma|Y_1^*U_1]
        \\
        &\le \frac{(N-1)\hat{N}}{\tilde{N}\exp(\gamma)}\label{ee4}
        \end{align}
        where $P_{S_1(n,\hat{n})Y_1^*U_1}=P_{U}P_{S|U}P_{Y_1^*|U}$ for $n\neq1$. Next, notice that
        \begin{align}
        \mathbb{P}[\mathcal{E}_1\cup\mathcal{E}_3]&\le\mathbb{E}\zeta(U_1,S^*,T^*)
        \\
        &=\int_0^1\mathbb{P}[\zeta(U_1,S^*,T^*)>v]{\rm d}v.\label{e73}
        \end{align}
        Note that $(S^*,T^*)$ has a complicated distribution since the coefficients are selected through the minimization of \eqref{ezeta}. To tackle this, let
        \begin{align}\label{eft}
        \mathcal{F}_v:=\{(u,s,t)\in\mathcal{U}\times\mathcal{S}\times\mathcal{T}:\zeta(u,s,t)\le v\}.
        \end{align}
        Then we can bound
        \begin{align}
        &\quad\mathbb{P}[\zeta(U_1,S^*,T^*)>v]
        \nonumber
        \\
        &=\mathbb{P}[\bigcap_{1\le\hat{n}\le \hat{N},1\le\hat{l}\le \hat{L}}
        \{(U_1,S_1(1,\hat{n}),T_1(1,\hat{l}))\notin\mathcal{F}_v\}]
        \\
        &\le \mathbb{P}[(U,S,T)\notin \mathcal{F}_v\cup\mathcal{S}]\nonumber
        \\
        &\quad+\frac{\min\{\hat{N},\hat{L}\}-1}{\hat{N}\hat{L}(\exp(-\gamma)-\exp(-2\gamma))}
        +e^{-\exp(\gamma)}\label{em}
        \end{align}
        where we invoked Lemma~\ref{lem5} in \eqref{em}, and defined the set
        \begin{align}
        \mathcal{S}:=\{(u,s,t):\imath_{S;T|U}(s;t|u)>\log\hat{N}\hat{L}-2\gamma\}.
        \end{align}
        Combining \eqref{e73} and \eqref{em} we have
        \begin{align}
        &\quad\mathbb{P}[\mathcal{E}_1\bigcup\mathcal{E}_3]
        \nonumber
        \\
        &\le \int_0^1\mathbb{P}[(U,S,T)\notin \mathcal{F}_v\cup\mathcal{S}]
        {\rm d}v
        \nonumber
        \\
        &\quad+\frac{\min\{\hat{N},\hat{L}\}-1}{\hat{N}\hat{L}(\exp(-\gamma)-\exp(-2\gamma))}
        +e^{-\exp(\gamma)}\label{e77}
        \end{align}
        but by definitions \eqref{e_56} and \eqref{eft},
        \begin{align}
        &\quad\int_0^1\mathbb{P}[(U,S,T)\notin \mathcal{F}_v\cup\mathcal{S}]{\rm d}v
        \nonumber
        \\
        &=\int_0^1\mathbb{E}1\{(U,S,T)\in\mathcal{F}_v^c\cup\mathcal{S}\}{\rm d}v
        \label{e_79}
        \\
        &=\mathbb{E}[1\{(U,S,T)\in\mathcal{S}\}+\zeta1\{(U,S,T)\notin\mathcal{S}\}]
        \\
        &=\mathbb{P}[(U,S,T)\in\mathcal{S}]
        \nonumber
        \\
        &\quad+\int{\rm d}P_{UST}P_{Y_1Y_2|UST}(\mathcal{Y}_{UST})1\{(U,S,T)\notin\mathcal{S}\}
        \\
        &=\mathbb{P}[(U,S,T)\in\mathcal{S}]
        \nonumber
        \\
        &\quad+\int{\rm d}P_{Y_1Y_2UST}1\{(Y_1,Y_2)\in\mathcal{Y}_{UST}
        \cup(U,S,T)\notin\mathcal{S}\}
        \\
        &=\mathbb{P}[(U,S,T)\in\mathcal{S}\cup(Y_1,Y_2)\in\mathcal{Y}_{UST}]
        \\
        &=\mathbb{P}[\{\imath_{S;T|U}(S;T|U)>\log\hat{N}\hat{L}-2\gamma\}
        \nonumber
        \\
        &\quad\cup\{\imath_{US;Y_1}(US;Y_1)\le \log M\tilde{N}+\gamma\}
        \nonumber\\
        &\quad\cup\{\imath_{UT;Y_2}(UT;Y_2)\le \log M\tilde{L}+\gamma\}
        \nonumber
        \\
        &\quad\cup\{\imath_{S;Y_1|U}(S;Y_1|U)\le \log \tilde{N}+\gamma\}
        \nonumber
        \\
        &\quad\cup\{\imath_{T;Y_2|U}(T;Y_2|U)\le \log \tilde{L}+\gamma\}].\label{e80}
    \end{align}
    Finally, the proof is accomplished by substituting \eqref{e80} into \eqref{e77} and then applying the union bound with \eqref{ee2}, \eqref{ee4} and \eqref{e77}.
\end{itemize}
\end{proof}

\begin{rem}
In the i.i.d.~setting Theorem~\ref{thm1} readily gives the following achievable region
\begin{align}
         \left\{
         \begin{array}{rl}
         R_{ii}&\le R_i
         \\
         \hat{R}_1+\hat{R}_2&\ge I({\sf U}_1;{\sf U}_2|{\sf U}_0)
         \\
         \sum_i(R_i-R_{ii})+R_0+R_{ii}+\hat{R}_i&\le I({\sf U}_0{\sf U}_i;{\sf Y}_i)
         \\
         R_{ii}+\hat{R}_i&\le I({\sf U}_i;{\sf Y}_i|{\sf U}_0)
         \end{array}
                  \right.
\end{align}
for some $R_{ii},\hat{R}_i>0$, $i=1,2$, $P_{{\sf U}_0{\sf U}_1{\sf U}_2}$ and function $x:\mathcal{U}_0\times\mathcal{U}_1\times\mathcal{U}_2\to\mathcal{X}$.
Fourier-Motzkin elimination gives the same region as in \cite{liang2007rate}; see also \cite[Theorem~8.4]{el2011network}.
\end{rem}

%

\section{Discussion}\label{sec_dis}
In contrast to the one-shot mutual covering lemma, a one-shot version of mutual \emph{packing} lemma \cite[Lemma~12.2]{el2011network} holds trivially, because the key step union bound in the proof of \cite[Lemma~2]{verdu2012non} does not require independence among the pairs.
\begin{lem}
Fix $(P_{XY},M,N,\gamma)$, then
\begin{align}
\mathbb{P}[\max_{m,n}\imath_{X;Y}(X_m;Y_n)\ge \log MN+\gamma]\le \exp(-\gamma)
\end{align}
where for all $1\le m\le M$, $1\le n\le N$, $P_{X_mY_n}=P_X\times P_Y$.
\end{lem}This can be used to derive a one-shot version of Berger-Tung inner bound without time sharing or stochastic decoders.

\section*{Acknowledgment}
This work was supported by NSF under Grants CCF-1350595, CCF-1116013, CCF-1319299, CCF-1319304, and the Air Force Office of Scientific Research under Grant FA9550-12-1-0196.

\bibliographystyle{ieeetr}
\bibliography{refm}

\end{document}